\documentclass[graybox]{svmult}

\usepackage{mathptmx}       
\usepackage{helvet}         
\usepackage{courier}        
\usepackage{type1cm}        
\usepackage{makeidx}         
\usepackage{graphicx}        
\usepackage{multicol}        
\usepackage[bottom]{footmisc}

\usepackage{amsmath,amssymb}
\usepackage{graphics,epsfig,color}
\usepackage{booktabs,siunitx}  
\usepackage[mathscr]{euscript}
\usepackage{verbatim}
\usepackage{hyperref,color}
\usepackage{wasysym}
\usepackage[dvipsnames]{xcolor}
\usepackage[framemethod=default]{mdframed}
\usepackage[Lenny]{fncychap}

\usepackage{tikz}

\usepackage[all]{xy}
\xyoption{all}

\newcommand{\mf}[1]{{\mathfrak #1}}
\newcommand{\e}{\emph}

\newcommand{\us}{\underset}
\newcommand{\os}{\overset}
\newcommand{\mc}[1]{{\mathcal #1}}
\newcommand{\bb}[1]{{\mathbb #1}}

\renewcommand{\phi}{\varphi}

\renewcommand{\tilde}{\widetilde}

\makeindex             

\begin{document}
	
	\title*{Geometrical structures of the instantaneous current and their macroscopic effects: vortices and perspectives in non-gradient models}
	\author{Leonardo De Carlo}
	\institute{Leonardo De Carlo \at Instituto Superior Tecnico,
		Avenida Rovisco Pais 1,   1000-189 Lisboa,  \email{leonardo_d3_carlo@protonmail.com}}
	
	\titlerunning{Geometrical structures of the instantaneous current and their macroscopic effects}
	
	%
	%
	\maketitle

	\abstract{First we discuss  the definition  of the instantaneous current in interacting particle systems,  in particular   in mass-energy systems and we point out its role in the derivation of the hydrodynamics.  Later we present some geometrical structures of the instantaneous current when the rates of stochastic models satisfy a very common symmetry. These structures give some new idea in non-gradient models and show a new phenomenology in  diffusive interacting particle systems. Specifically, we introduce models with vorticity and present some new perspectives on the link between the Green-Kubo's formula and the hydrodynamics of non-gradient models.}

	\smallskip
	
	
	\section{Introduction and Results}
	\label{sec:1}
	When many interacting particles are modeled by Newton's equations the rigorous derivation of hydrodynamics equations, consisting in some PDEs and describing the evolution of thermodynamic quantities,  is often a too optimistic programme, mainly  because of the lack of good ergodic property of the system. To overcome the mathematical problem two assumptions are traditionally made: or modeling the problem with a  stochastic microscopic evolution or assuming a low density of particles.  In the present framework we are interested in the first assumption and we are not having a complete rigorous point of view. For a rigorous and didactic treatment traditional  references  are \cite{KL99,Spohn}. The microscopic dynamics consists of random walks of particles on a lattice $ V_N $ that are  constrained to some rule expressing the local interaction, these are the so-called \textit{interacting particle systems} introduced by Spitzer \cite{Spit70}.

	\medskip
	
	In this paper we focus on the instantaneous current which is the bridge from the microscopic description to the macroscopic description of interacting particle systems. In section \ref{sec:definitions} we give some definitions  that we will use trough all the paper. In sections \ref{sec:part models} and \ref{sec:en models}  we present the models and describe the instantaneous current, in particular its definitions it is clarified in  in mass-energy systems like KMP \cite{DG2,KMP}. In section \ref{disch}  we recall the functional Hodge decomposition obtained in \cite{DG1} in dimension one and two  and we apply it to some interacting particle models. The expert reader can skip the first four sections and section 6, where well known notions of the literature are presented with a general flavour, and refer these sections  just for notation if necessary. 
	
	In the work the attention is  on diffusive models. Some new models with vorticity are introduced in section \ref{sec:slv}. After reviewing the qualitative theory of scaling limits in diffusive systems in section \ref{sec:SL}, in section \ref{sec:slv},  for the first time we study the macroscopic consequences of this decomposition. This leads us to some new phenomenology in particle systems, that is we show in a non-rigorous way that the hydrodynamics of the macroscopic current can present zero divergence terms that are not observed in the hydrodynamics of the density. This extend the usual Fick's law \eqref{eq:hdconJ} to the a new picture \eqref{eq:JasymD} where the diffusion matrix is  a positive non-symmetric matrix.

	In diffusive non gradient systems a derivation of the hydrodynamics with an explicit diffusion coefficient is an open problem.  The relative PDEs are in term of a variational expression of the diffusion coefficient equivalent to the Green-Kubo's formula, see \cite{N,Spohn}. In the last section \ref{sec:GK}, we try to give some perspectives coming from our Hodge decomposition. We give a possible explicit description of the minimizer of Green-Kubo's formula using our functional Hodge decomposition and describe a scheme that connects this minimizer with an explicit hydrodynamics.
	
	\section{Definitions}\label{sec:definitions}
	
	Interacting particle systems are stochastic models evolving  on a lattice along a continuous time Markov dynamics. 
	For the purposes of the paper, we are going to consider only periodic boundary conditions for the lattice where particles move, i.e. the set of vertices $ V_N $ of the lattice will be the  \e n-dimensional discrete torus $ \mathbb{T}^n_N = \mathbb{Z}^n/N \mathbb{Z}^n $ or $  \mathbb{T}^n_{\varepsilon} = \varepsilon\mathbb{Z}^n/N \mathbb{Z}^n $, where $ \varepsilon=1/N $ along the space scale we want to consider. We denote with   $\mathcal E_N$ the set of all couples of vertices $\{x,y\}$ of $ V_N$ such that $ y=x\pm \delta\, e_i $ where $ e_i $ is the canonical versor in $ \mathbb{Z}^n $ along the direction $ i $ and $ \delta$ is equals to $ 1 $ on $ \mathbb{T}^n_N $ and to $ 1/N $ on $ \mathbb{T}^n_\varepsilon$. The elements of $ \mc E_N $ are named  non-oriented edges or simply edges. In this way we have an non-oriented graph $ ( V_N,\mc E_N) $. To every non-oriented graph $( V_N, \mathcal E_N)$ we associate canonically an oriented graph $( V_N,E_N)$ such that the set of oriented edges $E_N$ contains all the ordered pairs $(x,y)$ such that $\{x,y\}\in\mathcal E_N$. Note that if $(x,y)\in E_N$ then also $(y,x)\in E_N$. If $e=(x,y)\in E_N$ we denote
	$e^-:=x$ and $e^+:=y$ and we call $ \mf e:=\{x,y\} $ the  non-oriented edge.  
	
	The microscopic configurations of our particle models are given by the collection of variables $ \eta(x) $ representing the number of particles, the energy or mass at $ x\in V_N $ along the model. When the variables $ \eta(x) $ are discrete we interpret them as number particles and  when continuous as  mass-energy. Calling $\Sigma $ the state space at $ x $ we define   the configuration state space as $\Sigma_N:= \Sigma^{ V_N}$. 
	The microscopic dynamics is a Markov process $ \{\eta_t\}_{t\in\mathbb{R}} $  where particles or masses   interact along  rules encoded in the generator  $ \mc{L}_N $, i.e.
	\begin{equation}\label{eq:gen}
	\mc{L}_Nf(\eta)=\sum_{\eta'\in\Sigma_N}c(\eta,\eta')[f(\eta')-f(\eta)],
	\end{equation}
	where $ f $ is an observable and $ c(\eta,\eta') $ the transition rates from $ \eta $ to $ \eta' $.
	
	Let $ \tau_z $ be the shift by $ z $ on $ \mathbb{Z}^n $ defined by the relation $ \tau_z\eta(x):=\eta(x-z) $ with $ z\in\mathbb{Z}^n $ and for a function $ h:\eta\to h(\eta)\in \mathbb{R} $ we define $\tau_z h(\eta):=h(\tau_{-z}\eta)  $, moreover for a domain $ B\subseteq V_N $ we define  $ \tau_z B:=B+z $. A function $h:\Sigma_N\to \mathbb R$ is called \e{local} if it depends only trough the configuration in  a finite domain $B\subset V_N $ denoted $ D(f) $. Let $ [\cdot]_+  $ be the positive part function.

	\section{Particle models and instantaneous current}\label{sec:part models}
	
	We treat only  nearest neighbour conservative dynamics, that is \eqref{eq:gen} becomes

	\begin{equation}\label{eq:discdyn}
	\mathcal L_N f(\eta)=\sum_{(x,y)\in E_N}c_{x,y}(\eta)\left(f(\eta^{x,y})-f(\eta)\right),\,\,\,\,\eta^{x,y}(z):=\left\{
	\begin{array}{ll}
	\eta(x)-1 & \textrm{if}\ z=x \\
	\eta(y)+1 & \textrm{if}\ z=y\\
	\eta(z) & \textrm{if}\ z\neq x,y
	\end{array}
	\right..
	\end{equation}
	We  study  \emph{translational covariant models}, i.e.  $c_{x,x\pm e^{(i)}}(\eta)=\tau_xc_{0,\pm e^{(i)}}(\eta)$   $\,\,\forall x\in V_N $.

	\subsection{Exclusion process and the 2-SEP}\label{ss:expr} In an \emph{exclusion process }particles move according to a conservative dynamics of  independent random walks with the exclusion rule that there cannot be more than one particle in a single lattice site (hard core interaction).
	The  rates of \eqref{eq:discdyn} have the general form
	\begin{equation}\label{eq:EPrate}
	c_{x,y}(\eta)=\eta(x)(1-\eta(y))\tilde c_{x,y}(\eta),
	\end{equation}
	where $ \tilde c_{x,y}(\eta) $ is the jump rates when  $ \eta $ has a particle in $ x $ and an empty site in $ y $.
	
	The next example of \eqref{eq:discdyn} is the \emph{2-SEP} (2-simple exclusion process), in this model the interaction is simply hardcore but  in every site there can be at most $ 2 $ particles. The state space is $ \Sigma_N=\{0,1,2\} $ and  the dynamics is defined by 
	\begin{equation}\label{eq:2SEPgen}
	\mathcal L_N^\textrm{2-SEP}f(\eta)=\sum_{(x,y)\in E_N}c_{x,y}(\eta)\left(f(\eta^{x,y})-f(\eta)\right),\,\,\, c_{x,y}(\eta)=\chi^+(\eta(x))\chi^-(\eta(y)),
	\end{equation}
	where $\chi^+(\alpha)=1$ if $\alpha >0$ and zero otherwise while $\chi^-(\alpha)=1$ if $\alpha<2$ and zero otherwise.

	\subsection{Instantaneous  current in particle systems}\label{ss:Istcurrpart}
	
	In interacting particle systems there are deep underlying   geometrical structures that  reflects in the hydrodynamics of lattice models   as we will discuss later, see also \cite{DG1}. The basis is the fact that the instantaneous current is a discrete vector field and closely related to a microscopic mass conservation law leading to the hydrodynamics.

	\begin{definition}
		A discrete vector field is a function $ \phi: E_N \to\mathbb{R} $ that is \emph{antisymmetric}, i.e.
		$ \phi(x,y) = - \phi(y,x) $ for any $ (x,y) \in E_N  $.
	\end{definition}
	The \e{instantaneous current} for our particle models is defined as
	\begin{equation}\label{eq:istcur}
	j_\eta(x,y):=c_{x,y}(\eta)-c_{y,x}(\eta)\,,
	\end{equation}
	which is a discrete vector field for each fixed configuration $ \eta $. The intuitive interpretation of the instantaneous current is the rate at which particles cross the bond $(x,y)$. Let $\mathcal N_t (x,y)$ be the number of particles that jumped from site $x$ to site $y$ up to time $t$. The \e{current flow} across the bond $(x,y)$ up to time $t$ is defined as
	\begin{equation}\label{eq:currver}
	J_t (x,y):=\mathcal N_t (x,y)-\mathcal N_t (y,x)\,.
	\end{equation}
	This is a discrete vector field ($J_t (x,y)=-J_t (y,x) $) depending on the trajectory $ \{\eta_t\}_t $.
	Between the instantaneous current $ j_\eta(x,y) $ and the current flow $ J_t (x,y) $ there is a strict  connection given by the key observation (see for example \cite{Spohn} section 2.3 in part II) that
	\begin{equation}\label{eq:mart}
	M_t (x,y)=J_ t (x,y)-\int_0^tj_{\eta(s)}(x,y)ds
	\end{equation}
	is a martingale. This allows to treat the difference between $ J_ t (x,y)  $ and the integral $ \int_0^t ds\, j_{\eta(s)}(x,y) $ as a  microscopic fluctuation term. It also gives a more physical definition of  $ j_\eta(x,y) $ as follows. Consider an initial configuration $ \eta_0=\eta $, the explicit expression of the instantaneous current can be defined as 
	\begin{equation}\label{eq:flow/t}
	j_\eta(x,y):=\us{ t\to0}{\lim}\frac{{\mathbb{E}^{\eta}(J_{t}(x,y))}}{t}.
	\end{equation}
	The expectation is $ \mathbb{E}^{\eta}(J_t(x,y))=\int\mathbb{P}^\eta(d\{\eta_t\}_t)J_t(x,y) $, where the integration is over all trajectories $ \{\eta_t\}_t $ starting from $ \eta $ at time 0 and $ \mathbb{P}^\eta $ the probability induced by the Markov process.  For a trajectory $\{\eta_t\}_t $ the probability to observe more than one jump goes like $ O(t^2) $, then it is negligible since we are interested in an infinitesimal time interval. Since 
	$ c(\eta,\eta')= \us{t\to 0}{\lim}\frac{\mathbb{P}^\eta(\eta_t=\eta')}{ t}=\us{t \to 0}{\lim}\frac{p_{ t}(\eta,\eta')}{t} $, where $ p_t (\eta,\eta') $ are the transition probability,  when  $ t $  goes to zero $ J_t(x,y) $  takes value $ +1 $ if a jump from $ x $ to $ y $ happens,  $ -1 $ in the opposite case and  $ 0 $ in the other cases. So the current defined in  \eqref{eq:flow/t}  becomes $ j_\eta(x,y)= c_{x,y}(\eta)-c_{y,x}(\eta) $ as in \eqref{eq:istcur}.
	
	The discrete divergence for a discrete vector field $ \phi $ on $ E_N $ is  $ \nabla\cdot\phi(x):=\us{y\sim x}{\sum}\phi(x,y) $, where the sum is on the nearest neighbours $ y\sim x $ of $ x $. For convenience of notation, we will use the symbol $ \nabla\cdot $ both for the discrete case and the continuous one, therefore we recommend to the reader to pay attention about this.  The local microscopic conservation law of the number of particles is then given by
	\begin{equation}\label{eq:numcon1}
	\eta_t(x)-\eta_0(x)+\nabla\cdot J_t(x)= 0.
	\end{equation}
	Using \eqref{eq:mart} in \eqref{eq:numcon1} we get
	\begin{equation}\label{eq:numcon2}
	\eta_t(x)-\eta_0(x)+\int_0^t ds\, \nabla\cdot j_s(x)+\nabla\cdot M_t(x)=0.
	\end{equation}
	
	We can deduce that  at the equilibrium, that is when for a measure $ \mu_N $ on $ \Sigma_N $ the \emph{detailed balance condition} 
	is true, i.e. $
	\mu_N (\eta) c(\eta,\eta^{x,y})=\mu_N (\eta^{x,y}) c(\eta^{x,y},\eta) 
	$  for all $ (x,y)\in E_N $, the average flow $ \mathbb{E}^\eta_{\mu_N} (J_t(x,y)) $ is constantly zero, where  the subscript $ \mu_N $   indicates  the average respect to the equilibrium measure $ \mu_N $.   For a small time interval $ \Delta t $  from \eqref{eq:mart}, \eqref{eq:flow/t} and the detailed balance we have $ \mathbb{E}^\eta_{\mu_N}(J_{\Delta t}(x,y)) \sim \bb E_{\mu_N} (j_\eta(x,y)) {\Delta t}=0 $. Since this is true for any time interval $ \Delta t $ and the current flow $ J_t(x,y) $ is additive we conclude that $ \mathbb{E}^\eta_{\mu_N}(J_{t}(x,y))=0 $. More generally for a \emph{stationary measure} $ \mu_N $, that is $ \mu_N (\mathcal L_N f)=0 $ for any $ f $,   we have that 
	\begin{equation}\label{eq:asJ=asj}
	\bb E^\eta_{\mu_N}  (J_t(x,y))=\bb E_{\mu_N} (j_\eta(x,y)) t.
	\end{equation}

	\begin{remark}\label{re:trasrate}
		For a \emph{translational covariant model}, i.e. $ c_{x,y}(\eta)=c_{x+z,y+z}(\tau_z\eta) $ for any $ z\in V_N $, then the instantaneous current is translational covariant too, namely  it satisfies the symmetry relation
		\begin{equation}\label{eq:jinv}
		j_\eta(x,y)=j_{\tau_z\eta}(x+z,y+z).
		\end{equation}
	\end{remark}

	\section{Energy-mass models}\label{sec:en models}
	In this section we adapt the concepts of the previous section to the continuous case, where we consider  models that   exchange  continuous quantity between sites. The lattice variables are interpreted as energy or mass along the context and the configuration is denoted with $ \xi=\{\xi(x)\}_{x\in V_N} $.
	The first model to be described is  the most famous model of this class, namely the Kipnis-Marchioro-Presutti (KMP) model \cite{KMP}.
	
	\subsection{KMP model and generalization, dual KMP, gaussian model} The \emph{KMP dynamics} is a generalized stochastic lattice gas on which energies  or masses are associated to oscillators at the vertices $ V_N $ .    The stochastic evolution is of the type
	\begin{equation}\label{eq:genKMP}
	\mathcal L_N f(\xi) =\sum_{\{x,y\}\in\mc{E}_N} \mc{L}_{\{x,y\}}f(\xi) \,,\text{ with }
	\end{equation}
	\begin{equation}\label{eq:KMPbulkgen}
	\mc{L}_{\{x,y\}}f(\xi):=\int_{-\xi(y)}^{\xi(x)}\frac{dq}{\xi(x)+\xi(y)}\big[f(\xi-q\left(\varepsilon^x-\varepsilon^y\right))-f(\xi)\big]\,.
	\end{equation}
	where $\varepsilon^x=\left\{\varepsilon^x(y)\right\}_{y\in  V_N}$ is the configuration of mass with all the sites different from $x$ empty and having unitary mass at site $x$,  this means that $\varepsilon^x(y)=\delta_{x,y}$ where $\delta$ is the Kronecker symbol. Formula \eqref{eq:KMPbulkgen} define the model as a uniform distributed random current model. 

	The dynamics \eqref{eq:KMPbulkgen}  can be generalized substituting the uniform distribution on $[-\xi(y),\xi(x)]$ for
	a different probability measure (or just positive measure) $\Gamma_{x,y}^\xi(dq)$, i.e.
	\begin{equation}\label{eq:KMPgeneralized}
	\mc{L}_{\{x,y\}}f(\xi):=\int\Gamma_{x,y}^\xi(dq)[f(\xi-q\left(\varepsilon^x-\varepsilon^y\right))-f(\xi)\big]\,
	\end{equation}
	with the symmetry $\Gamma^\xi_{x,y}(q)=\Gamma^\xi_{y,x}(-q)$ so that 
	\eqref{eq:genKMP} is a sum over unordered edges.
	When considering  a discrete state space, a natural choice for $ \Gamma_{x,y}^\xi(dq) $ in \eqref{eq:KMPgeneralized} is the discrete uniform distribution on the integer points in $[-\xi(y),\xi(x)]$. This means that if $\xi$ is a configuration of mass assuming only integer values then
	\begin{equation}\label{eq:KMPdrate}
	\Gamma_{x,y}^\xi(dq)=\frac{1}{\xi(x)+\xi(y)+1}\sum_{i\in [-\xi(y),\xi(x)]}\delta_i(dq)
	\end{equation}
	where $\delta_i(dq)$ is the delta measure  at $i$ and the sum is over the integer values belonging to the interval.  This is exactly the dual model of KMP \cite{KMP} called also \emph{KMPd}.

	Another interesting model could be the following   \emph{Gaussian model}. In this case the interpretation in terms of mass is missing since the variables can assume also negative values and it could be interpreted as a charge model. The bulk dynamics is defined by a distribution of current having support on all the real line
	\begin{equation}\label{eq:gaussrate}
	\Gamma_{x,y}^\xi(dq)=\frac{1}{\sqrt{2\pi \gamma^2}}e^{-\frac{\left(q-\frac{(\xi(x)-\xi(y))}{2}\right)^2}{2\gamma^2}}dq\,.
	\end{equation}

	\subsection{Weakly asymmetric energy-mass models}\label{subsec:WAm}  We consider dynamics perturbed by a space and time dependent discrete external field $ \mathbb{F} $ defined as follows.
	Let $ F:\mathbb T^n\to\mathbb{R}^n $ be a smooth vector field with components $ F(x)=(F_1,\dots,F_n) $, describing the force acting on the masses of the systems. We associate to $ F $ a discrete vector field $ \mathbb{F}(x,y) $ defined by 
	\begin{equation}\label{eq:disF}
	\mathbb{F}(x,y)=\int_{(x,y)}F(z)\cdot dz,
	\end{equation}
	$(x,y)$ is an oriented edge  and the integral is a line integral
	that corresponds to the work done by the vector field $F$ when a particle moves from $x$ to
	$y$. So we think about $ \mathbb F(x,y) $ as work done per particle. 
	We want to change the  random distribution \eqref{eq:KMPgeneralized}  of the current on each bond according to a perturbed measure $ \Gamma^{\mathbb F} $, that is
	\begin{equation}\label{eq:bulkxyE}
	\mc{L}_{\{x,y\}}^{\mathbb F}f(\xi):=\int\Gamma^{\xi,\mathbb F}_{x,y}(dq)\big[f(\xi-q\left(\varepsilon^x-\varepsilon^y\right))-f(\xi)\big], \,\,\,\,\,\, \Gamma^{\xi,\mathbb F}_{x,y}(dq)=\Gamma^{\xi}_{x,y}(dq)e^{\frac{\mathbb F(x,y)}{2}q}\,.
	\end{equation}
	The effect of an external field is modelled by perturbing the rates and giving a net drift toward a specified direction. When  the size of  $ |y-x| $ is of order $ 1/N $ we obtain a \emph{ weakly asymmetric model},  the discrete vector field \eqref{eq:disF} is of order $ 1/N $ too and  the hydrodynamics is studied  considering a perturbative expansion of $ \Gamma^{\xi,\mathbb F}_{x,y}(dq) $ for the orders that will give  a macroscopic effect. We will see that for weakly asymmetric diffusive models this expansion is necessary up to the order two. If $ F=-\nabla H $ is a gradient vector field, then $ \mathbb F(x,y)=H(x)-H(y) $ and  $\Gamma^{\xi,\mathbb F}_{x,y}(dq)=\Gamma^{\xi}_{x,y}(dq)e^{(H(x)-H(y))q} $.  

	By the symmetry of the measure $\Gamma$ and the antisymmetry of the discrete vector field $\mathbb F$ we have that $\Gamma^{\xi,\mathbb F}_{x,y}(q)=\Gamma^{\xi,\mathbb F}_{y,x}(-q)$ and we can define the generator considering sums over unordered bonds
	\begin{equation}\label{eq:fullKMPgen}
	\mc{L_N}f(\xi)=\sum_{\{x,y\}\in\mc{E}_N}\mc{L}_{\{x,y\}}^{\mathbb F}f(\xi).
	\end{equation}

	\subsection{Instantaneous current of energy-mass systems}\label{subsec:Ie-mc} Here we adapt the definition of instantaneous current to the formalism of the interacting nearest neighbour energy-mass models. The generator is  \eqref{eq:bulkxyE}, the case $ \mathbb{F}=0 $ is treated as a subcase  and we omit the index when the external field is zero. 
	The \e{instantaneous current} for the bulk dynamics is defined as
	\begin{equation}\label{eq:istce-m}
	j^{\mathbb{F}}_\xi(x,y):=\int \Gamma^{\xi,\mathbb{F}}_{x,y}(dq)q\,.
	\end{equation}
	Its interpretation is the rate at which masses-energies cross the bond $ (x,y) $ and it is still a discrete vector field.
	The \e{current flow} now is indicated with $\mathcal J_t(x,y)$ and it is the net total amount of mass-energy that has flown from $x$ to $y$ in the time window $[0,t]$. It can be defined as sum of all the differences between  the mass-energy measured in $ x  $ before and after of every jump on the bond $ \{x,y\} $. Let $ \tau_i $ be the time of the $ i-th $ jump on the bond $ \{x,y\} $ for some $ i $, we write the current flow as follows
	\begin{equation}\label{eq:e-mflow1}
	\mc{J}_t(x,y):=\sum_{\tau_i:\tau_i\in[0,t]} J_{\tau_i}(x,y) \,,
	\end{equation}
	where  $ J_{\tau}(x,y) $ is the \emph{present flow} defined as the current flowing from $ x $ to $ y $   jump time  $ \tau $
	\begin{equation}\label{eq:eq_e-mflow2}
	J_{\tau}(x,y):=\lim_{h\downarrow 0}\xi_{\tau-h}(x)-\lim_{h\downarrow 0}\xi_{\tau+h}(x).
	\end{equation}
	Defining $ J_{\tau}(y,x):=\lim_{h\downarrow 0}\xi_{\tau-h}(y)-\lim_{h\downarrow 0}\xi_{\tau+h}(y) $, the flow $ \mc{J}_{t}(x,y) $ is still an anti-symmetric vector field depending on the trajectory $ \{\xi_t\} $, i.e. $ J_{\tau}(y,x):=-J_{\tau}(x,y) $.
	As in the particles case  $ \mc{J}_t(x,y) $ is a function on the path space, while the instantaneous current $ j^{\mathbb{F}}_\xi(x,y) $ is a function on the configuration space and  the difference 
	\begin{equation}\label{eq:massmart}
	M_t(x,y)=\mc{J}_t(x,y)-\int^t_0 ds\,j^{\mathbb{F}}_{\xi(s)}(x,y).
	\end{equation}
	is a martingale.  
	Repeating what we did in  subsection \ref{ss:Istcurrpart} (with a formalism suitable to energy-mass models)  the instantaneous current \eqref{eq:istce-m} can be obtained as 
	\begin{equation}\label{eq:massflow/t}
	j^{\mathbb{F}}_\xi(x,y):=\us{ t\to0}{\lim}\frac{{\mathbb{E}^{\xi}(\mc{J}_{t}(x,y))}}{t}.
	\end{equation}

	As we did in subsection \ref{ss:Istcurrpart}  from the local discrete conservation of the mass-energy  $ \xi_t(x)-\xi_0(x)+\nabla\cdot \mc{J}_t(x)= 0 $ we have 
	\begin{equation}\label{eq:masscon}
	\xi_t(x)-\xi_0(x)+\int_0^t ds\, \nabla\cdot j^{\mathbb{F}}_{\xi(s)}(x)+\nabla\cdot M_t(x)=0.
	\end{equation}
	The microscopic fluctuation \eqref{eq:massmart} has mean zero and   \eqref{eq:asJ=asj} can be obtained similarly  to conclude that the average currents are zero in the equilibrium case, i.e. when detailed balance conditions (DBC) hold.  
	
	The natural scaling limit for this class of processes is the diffusive one, where the rates have to be multiplied by $N^2$ to get a non trivial scaling limit. So,  instead of \eqref{eq:massmart}, we will consider  in the  macroscopic theory the speeded up martingale   
	$
	M_t(x,y)=\mc{J}_t(x,y)-N^2\int^t_0 ds\,j^\mathbb{F}_{\xi(s)}(x,y).
	$
	
	\begin{example}
		For example the instantaneous current across the edge $(x,y)$ for the KMP process is given by
		\begin{equation}\label{eq:KMPic}
		\int_{-\xi(y)}^{\xi(x)}\frac{qdq}{\xi(x)+\xi(y)}=\frac 12\left(\xi(x)-\xi(y)\right)\,.
		\end{equation}
		This computation shows that the KMP model is of gradient type, see definition \eqref{grgr},  with $h(\xi)=-\frac{\xi(0)}{2}$. Also the KMPd is gradient with respect to the same function $h$.
	\end{example}
	
	\begin{example}
		For the weakly asymmetric KMP in the case of a constant external field $F=E$ in the direction from $x$ to $y$ the discrete field $ \mathbb{F}(x,y)$ is given by $E/N $ on $ \mathbb{T} ^n_\varepsilon$ and 
		\begin{equation*}\label{eq:KMPEgamma}
		\Gamma^{\xi,E}_{x,y}(q)=\frac{1+\frac{E}{N}q}{\xi(x)+\xi(y)}+o(N)
		\end{equation*}
		Then the instantaneous current is
		\begin{eqnarray}\label{eq:KMPEic}
		& & j_{\xi}^E(x,y)=\int_{-\xi(y)}^{\xi(x)}\Gamma^{\xi,E}_{x,y}(q)qdq\nonumber= \\
		& &\frac{2 N}{E(\xi(x)+\xi(y))}\left[e^{\frac {E}{2N}}\xi(x)
		+e^{-\frac {E}{2N}\xi(y)}\xi(y)-2\frac{e^{\frac {E}{2N}\xi(x)}-e^{-\frac {E}{2N}\xi(y)}}{E}\right]\nonumber = \\
		& &=\frac 12\big(\xi(x)-\xi(y)\big)+\frac {E}{N}6\big[\xi(x)^2+\xi(y)^2-\xi(x)\xi(y)\big]+o(N)\,.
		\end{eqnarray}
		The hydrodynamic behavior of the model under the action of an external field in the weakly asymmetric regime, i.e. when the external field $E$ is of order $1/N$, is determined by the first two orders in the expansion \eqref{eq:KMPEic}. In particular any perturbed KMP model having the same expansion as in \eqref{eq:KMPEic} will have the same hydrodynamics.
		
		While for the KMPd model we get
		{\small	{\begin{equation}\label{eq:KMPdEic}
				j_{\xi}^E(x,y)=\frac 12\big(\xi(x)-\xi(y)\big)+\frac{E}{N}{12}\big[2\xi(x)^2+2\xi(y)^2-2\xi(x)\xi(y)+3\xi(x)+3\xi(y)\big]+o(N)\,.
				\end{equation}}}
	\end{example}
	\section{Discrete Hodge decomposition in interacting particle systems}\label{disch}
	
	In the first section we defined the graph $ (V_N,\mathcal E_N) $. Now we  enter into the detail of the discrete mathematics we need to study the geometrical structures of the current.  We consider the case when the graph $(V_N,\mathcal E_N)$ is on $ \mathbb{T}^2_N $. 

	A sequence $(z_0,z_1,\dots ,z_k)$ of elements of $V_N$ such that $(z_i,z_{i+1})\in E_N$, $i=0,\dots k-1$, is called an oriented  path, or simply a \emph{path}. A \emph{cycle} $ C = (z_0,z_1,\dots ,z_k) $ is  a path  with distinct vertices except $z_0=z_k$ and it is defined as an equivalence class modulo cyclic permutations.
	If $C$ is a cycle and there exists an $i$ such that $(x,y)=(z_i,z_{i+1})$ we write $(x,y)\in C$. Likewise if there exists an $i$ such that $x=z_i$ we write $x\in C$. 
	A \emph{discrete vector field} $\phi$ on $(V_N,E_N)$ is a map $\phi:E_N\to \mathbb R$ such that $\phi(x,y)=-\phi(y,x)$.
	A discrete vector field is of \emph{gradient type} if there exists a function $h:V_N\to \mathbb R$ such that
	$\phi(x,y)=[\nabla h](x,y):= h(y)-h(x)$.
	The divergence of a discrete vector field $\phi$ at $x\in V_N$ is defined by
	\begin{equation}\label{eq:disdiv}
	\nabla\cdot \phi(x):=\sum_{y\,:\, \{x,y\}\in \mathcal E_N}\phi(x,y)\,.
	\end{equation}We call $\Lambda^1$ the $|\mathcal E_N|$-dimensional vector space of discrete vector fields. We endow $\Lambda^1$ with the scalar product
	\begin{equation}\label{sc}
	\langle \phi,\psi\rangle:=\frac 12\sum_{(x,y)\in E_N}\phi(x,y)\psi(x,y)\,, \qquad \phi,\psi\in \Lambda^1\,.
	\end{equation}

	We recall briefly the Hodge decomposition for discrete vector fields. We call $\Lambda^0$ the collection of real valued function defined on the set of vertices
	$\Lambda^0:=\{g\, :\, V_N\to \mathbb R\}$.  Finally we call $\Lambda^2$ the vector space of \emph{2-forms} defined on the faces of the lattice $\mathbb Z^2_N$. Let us define this precisely. An oriented face is for example an elementary cycle in the graph of the type $(x,x+e^{(1)}, x+e^{(1)}+e^{(2)},x+e^{(2)},x)$ . In this case we have an \emph{anticlockwise oriented face}. This corresponds geometrically to a square having vertices $x,x+e^{(1)},x+e^{(1)}+e^{(2)},x+e^{(2)}$ plus an orientation in the anticlockwise sense. The same elementary face can be oriented \emph{clockwise} and this corresponds to the elementary cycle $(x,x+e^{(2)}, x+e^{(1)}+e^{(2)},x+e^{(1)},x)$. If $f$ is a given oriented face we denote by $-f$ the oriented face corresponding to the same geometric square but having opposite orientation. A 2-form is a map $\psi$ from the set of oriented faces $F_N$ to $\mathbb R$ that is antisymmetric with respect to the change of orientation, i.e. such that $\psi(-f)=-\psi(f)$. The boundary $\delta\psi$ of  $\psi$ is a discrete vector field defined by
	\begin{equation}\label{defd}
	\delta\psi(e):=\sum_{f\,:\, e\in f}\psi(f)\,.
	\end{equation}
	Since a face is  a cycle the meaning of $e\in f$ has been just discussed above. Note that \eqref{defd} is a discrete orthogonal gradient,  the orthogonal gradient  $ \nabla^\perp f $ of a  smooth function $ f $ is defined as $(-\partial_y f,\partial_x f) $. In higher dimension this a discrete curl.

	By construction $\nabla\cdot \delta\psi=0$ for any $\psi$.
	The 2-dimensional \emph{discrete Hodge decomposition} is written as the direct sum
	\begin{equation}\label{dish2}
	\Lambda^1=\nabla \Lambda^0\oplus \delta\Lambda^2\oplus\Lambda^1_H\,,
	\end{equation}
	where the orthogonality is with respect to the scalar product \eqref{sc}. The discrete vector fields on $\nabla\Lambda^0$ are the gradient ones. The dimension of $\nabla\Lambda^0$ is $N^2-1$. The vector subspace $\delta\Lambda^2$ contains all the discrete vector fields that can be obtained by \eqref{defd} from a given 2-form $\psi$. The dimension of $\delta\Lambda^2$ is $N^2-1$. Elements of $\delta\Lambda^2$ are called \emph{circulations}. The dimension of $\Lambda_H^1$ is simply 2. Discrete vector fields in $\Lambda^1_H$ are called \emph{harmonic}. A basis in $\Lambda^1_H$ is given by the vector fields $\varphi^{(1)}$ and $\varphi^{(2)}$ defined by
	\begin{equation}\label{cenes}
	\varphi^{(i)}\left(x,x+e^{(j)}\right):=\delta_{i,j}\,, \qquad i,j=1,2\,.
	\end{equation}
	Given a vector field $\phi\in \Lambda^1$, we write
	$
	\phi=\phi^\nabla+\phi^\delta+\phi^H
	$
	to denote the unique splitting in the three orthogonal components. This decomposition can be computed as follows. The harmonic part is determined writing $\phi^H=c_1\varphi^{(1)}+c_2\varphi^{(2)}$ with The coefficients $c_i$  determined by
	$
	c_i=\frac{1}{N^2}\sum_{x\in V_N}\phi\left(x,x+e^{(i)}\right)\,.
	$
	To determine the gradient component $\phi^\nabla$ we need to determine a function $h$ for which $\phi^\nabla(x,y)=[\nabla h](x,y)=h(y)-h(x)$. This is done  by taking the divergence on both side of $ \phi=\phi^\nabla+\phi^\delta+\phi^H $ and obtaining the $h$ solving  the discrete Poisson equation $\nabla\cdot\nabla h=\nabla\cdot \phi$. The remaining component $\phi^\delta$ is computed just by difference $\phi^\delta=\phi-\phi^\nabla-\phi^H$.
	We refer to \cite{DC,GV} for a version of discrete calculus with cubic cells and to \cite{DKT} for a version of discrete calculus with simplexes.
	
	\medskip
	
	Given an oriented edge $e$ or an oriented face $f$ we denote respectively by $\mathfrak e$, $\mathfrak f$ the corresponding un-oriented edge and face. 
	Note that  both $f$ and $-f$ are associated with the same un-oriented face $\mathfrak f$.
	Given an oriented edge $e\in  E_N$ of the lattice there is only one anticlockwise oriented face to which
	$e$ belongs that we call it $f^+(e)$. There is also an unique anticlockwise face, that we call $f^-(e)$, such that $e\in -f^-(e)$
	(see Figure 1).
	
	It is useful to define $\tau_{\mathfrak f}$ for an un-oriented face $\mathfrak f$. If  $\mathfrak f = \{x,x+e^{(1)},x+e^{(2)},x+e^{(1)}+e^{(2)}\}$ then we define $\tau_{\mathfrak f}:=\tau_x$.
	For  $\mathfrak e=\{x,x+e^{(i)}\}$ we define $\tau_{\mathfrak e}:=\tau_x$. We  use also the notation $ f^\circlearrowleft $ for an anticlockwise face and $ f^\circlearrowright $ for a clockwise one.

	\begin{figure}[]
		\centering
		\hspace{0.1cm}
		\entrymodifiers={+<0.5ex>[o][F*:black]}
		\xymatrix@C=1.2cm@R=1.2cm{
			{}\ar@{.}[rrrr]\ar@{.}[dddd] & \ar@{.}[dddd] & {}\ar@{.}[dddd] & {}\ar@{.}[dddd] & {}\ar@{.}[dddd]\\
			{}\ar@{.}[rrrr] & {} & {} & {} & {}\\
			{}\ar@{.}[rrrr]_(.53){\textit{\normalsize y}} & {} & {} & {} & {}\\
			{}\ar@{.}[rrrr]_(.53){\textit{\normalsize x}}^(0.38){\textit{\tiny{$f^{+}(e)$}}} ^(0.63){\textit{\tiny{$f^{-}(e)$}}} & {} \ar@{}[ur]|{\textit{\Huge{$\circlearrowleft$}}} & {} \ar@{->}[u]|{\textit{}} & {} \ar@{}[ul]|{\textit{\Huge{$\circlearrowleft$}}} & {}\\
			{} \ar@{->}[u]^(.5){\textit{\normalsize $e^{(2)}$}}  \ar@{->}[r]_(.5){\textit{\normalsize $e^{(1)}$}} \ar@{.}_(-0.05){\text{\normalsize }}[rrrr] & {} & {} & {} & {}
		}
		\caption{\small{ On discrete two dimensional torus, given $(x,y)=e$  we draw the faces $f^{-}(e)$ and $f^{+}(e)$. }}
		\label{fig: due facce}
	\end{figure}
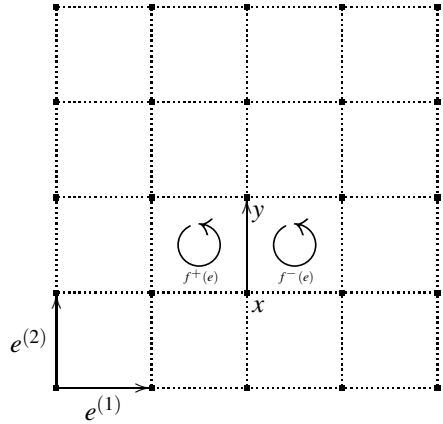

	\subsection{Functional discrete Hodge decomposition and lattice gases}
	\label{fdhhd}
	A relevant notion in the derivation of the hydrodynamic behavior for diffusive particle systems is the definition of gradient particle system. A particle system is called of \emph{gradient type} if there exists a local function $h$ such that
	\begin{equation}\label{grgr}
	j_\eta(x,y)=\tau_yh(\eta)-\tau_x h(\eta)\, \text{ for all } (\eta,(x,y))\in (\Sigma_N, E_N).
	\end{equation}
	The relevance of this notion is on the fact that the proof of the hydrodynamic limit for gradient systems is extremely simplified.
	Moreover for gradient and reversible models it is possible to obtain explicit expressions of the transport coefficients.
	
	Here we show that \eqref{grgr} is a subcase of  general geometrical structures for the instantaneous current. In next sections, we will try to understand the consequences of these structures in the hydrodynamic limits and how it could be useful in understanding the hydrodynamics of non-gradient models. We present a functional Hodge decomposition of
	\emph{translational covariant discrete vector fields}. This means vector fields $ j_\eta(x,y) $ depending on the configuration $\eta\in\Sigma_N$ and satisfying \eqref{eq:jinv}.
	Vector fields of the form \eqref{grgr} play the role of the gradient vector fields. Circulations will also be suitably defined in the context of particle systems.

	\subsection{The one dimensional case}
	
	On the one dimensional torus  $V_N$, we have the following theorem.
	
	\begin{theorem}\label{belteo}
		Let $j_\eta$ be a translational covariant discrete vector field. Then there exists a function $h(\eta)$ and a translational invariant function $C(\eta)$ such that
		\begin{equation}\label{imbr}
		j_\eta(x,x+1)=\tau_{x+1}h(\eta)-\tau_{x}h(\eta)+C(\eta)\,.
		\end{equation}
		The function $C$ is uniquely identified and coincides with
		\begin{equation}\label{prim}
		C(\eta)=\frac 1N\sum_{x\in  V_N} j_\eta(x,x+1)\,.
		\end{equation}
		The function $h$ is uniquely identified up to an arbitrary additive translational invariant function and coincides with
		\begin{equation}\label{formulah}
		h(\eta)=\sum_{x=1}^{N-1}\frac{x}{N}j_\eta(x,x+1)\,.
		\end{equation}
	\end{theorem}
	
	\begin{proof}
		
		The basic idea of the theorem is the usual strategy to construct the potential of a gradient discrete vector field plus a subtle use of the translational covariance of the model.
		For the details of the proof see \cite{DG1}.
	\end{proof}
	
	Observe that a one dimensional system of particles is of gradient type (with a possibly not local $h$)
	if and only if $C(\eta)=0$. This corresponds to say that for any fixed configuration $\eta$ then $j_\eta(x,y)$
	is a gradient vector field. This was already observed in \cite{BDGJL,N}.
	Now we compute  the decomposition \eqref{imbr} in some examples. Later we will discuss how it can be related to the hydrodynamics of non-gradient systems.

	\begin{example}\label{exa:diffSEP}
			On the one-dimensional discrete torus, the  symmetric exclusion process with rates $ c_{x,x+1}(\eta)=\eta(x)(1-\eta(x+1))[1+\alpha\eta(x-1)] $ and  $c_{x+1,x}(\eta)=\eta(x+1)(1-\eta(x))[1+\alpha\eta(x-1)]  $, with the constant $ \alpha\in(0,1) $,  is reversible with respect to the Bernoulli measure.  This is a non-gradient systems, expected to have a diffusive scaling limits, where the instantaneous current is given by
			\begin{equation}\label{eq:currdiffSEP}
			j_\eta(x,x+1)= (\eta(x)-\eta(x+1))+(\eta(x)-\eta(x+1))\alpha \eta(x-1).
			\end{equation} 
			Therefore its functional Hodge decomposition \eqref{imbr} is 
			{\small \begin{equation}\label{eq:hdiffSEP}
				h(\eta)=-\eta(0)+ \sum_{x=1}^{N-1}\frac{x}{N}(\eta(x)-\eta(x+1))\alpha \eta(x-1),\,\,\,C(\eta)= \sum_{x\in V_N}\frac{x}{N}(\eta(x)-\eta(x+1))\alpha \eta(x-1).
				\end{equation}}
	\end{example}

	\begin{example}[\emph{The 2-SEP}]\label{exa:2-SEP}
		The model we are considering is the 2-SEP, see its defintion in subsection \ref{ss:expr}. 
		We denote by  $D^\pm_\eta(x,x+1)$ the  local functions associated with the presence on the bond $(x,x+1)$ of what we call respectively a positive or negative discrepancy. More precisely $D^+_\eta(x,x+1)=1$ if $\eta(x)=2$ and $\eta(x+1)=1$ and zero otherwise. We have instead $D^-_\eta(x,x+1)=1$ if $\eta(x+1)=2$ and $\eta(x)=1$ and zero otherwise. We define also $D_\eta:=D^+_\eta-D^-_\eta$.
		The instantaneous current across the edge $(x,x+1)$ associated with the configuration $\eta$ is
		\begin{equation}\label{eq:curr2-SEP}
		j_\eta(x,x+1):=\chi^+(\eta(x))-\chi^+(\eta(x+1))+D_\eta(x,x+1).
		\end{equation}
		For this specific model formulas \eqref{prim} and \eqref{formulah} become\begin{equation}\label{eq:H2-SEP}
		h(\eta)=-\chi^+(\eta(0))+\sum_{x=1}^{N-1}\frac{x}{N}D_\eta(x,x+1),\,\,\,\,\,C(\eta)=\frac 1N\sum_{x\in V_N}D_\eta(x,x+1).
		\end{equation}
	\end{example}
	
	\begin{remark}\label{re:ja}
			Both formulas \eqref{eq:currdiffSEP} and \eqref{eq:curr2-SEP} are written in the form $ j_\eta(x,y)=j^h_\eta(x,y)+j^a_\eta(x,y) $, namely  they are given by the sum of a local gradient current $ j^h_\eta(x,y)=\tau_y h(\eta)- \tau_x h(\eta)$ and a single net contribution $ j^a_\eta (x,y)$ to the harmonic function $ C(\eta) $. We will refer  to $ j^a_\eta (x,y) $ as \emph{single harmonic contribution} on $ (x.y) $. In section \ref{sec:GK}, we will discuss that we think from this way of rewriting the current  it has  to start both the study of an explicit hydrodynamics for the case of non-gradient diffusive model and the computation of Green-Kubo's formula.
	\end{remark}
	
	Our decomposition is motivated by the study of diffusive models where the current can not be written in the gradient form \eqref{grgr}, but it can be computed also in not diffusive models when the hypothesis of theorem \eqref{belteo} hold. For example for the asymmetric simple exclusion process it is as follows.
	
	\begin{example}[\emph{ASEP}]
		The asymmetric simple exclusion process is characterized by the rates $c_{x,x+1}(\eta)=p\eta(x)(1-\eta(x+1))$ and $c_{x,x-1}(\eta)=q\eta(x)(1-\eta(x-1))$.
		Given a configuration of particles $\eta\in \Sigma$, we call $\mathfrak C(\eta)$ the collection of clusters of particles
		that is induced on $V_N$.
		A cluster $c\in \mathfrak C(\eta)$ is a subgraph of $(V_N,\mathcal E_N)$. Two sites $x,y\in V_N$ belong to the same cluster $c$ if
		$\eta(x)=\eta(y)=1$ and there exists an un-oriented path $(z_0,z_1, \dots ,z_k)$ such that $\eta(z_i)=1$ and
		$(z_i,z_{i+1})\in \mathcal E_N$. Given a cluster $c\in \mathfrak C$ we call $\partial^lc$ and $\partial^rc\in V_N$ respectively the first element on the left of the leftmost site of $ c $ and the rightmost  one. The decomposition \eqref{imbr} holds with
		\begin{equation}\label{caldaia}
		h(\eta)=\frac 1N\sum_{c\in \mathfrak C(\eta)}\left[p\partial^rc-q\partial^lc\right],\,\,\,\,\,C(\eta)=\frac{(p-q)\left|\mathfrak C(\eta)\right|}{N}.
		\end{equation}
		where $\left|\mathfrak C(\eta)\right|$ denotes the number of clusters.
	\end{example}

	\subsection{The two dimensional case}

	On the two dimensional torus $V_N$ the decomposition is as follows.
	
	\begin{theorem}\label{belteo2}
		Let $j_\eta$  be a covariant discrete vector field. Then there exist four functions $h,g,C^{(1)}, C^{(2)}$ on configurations of particles such that for an edge of the type $e=(x,x\pm e^{(i)})$ we have
		\begin{equation}\label{hodgefun2}
		j_\eta(e)=\big[\tau_{e^+}h(\eta)-\tau_{e^-}h(\eta)\big]+\big[\tau_{\mathfrak f^+(e)}g(\eta)-\tau_{\mathfrak f^-(e)}g(\eta)\big]\pm C^{(i)}(\eta)\,.
		\end{equation}
		The functions $C^{(i)}=\frac{1}{N^2}\underset{x\in V_N}{\sum} j_\eta (x,x+e^{(i)})$ are translational invariant and uniquely identified. The functions $h$ and $g$ are uniquely identified up to additive arbitrary translational invariant functions.
	\end{theorem}
	\begin{proof}
		see \cite{DG1}.
	\end{proof}
	
	We remark that the proof in \cite{DG1} is constructive, that is the function $ h(\eta),g(\eta) $ and $C^{(i)}(\eta)  $ have explicit expressions. In analogy to gradient systems we can say a particle system is  of \emph{circulation type} when  there exist a local function $g$ such that
	\begin{equation}\label{eq:cicrtype}
	j_\eta(e)=\tau_{\mathfrak f^+(e)}g(\eta)-\tau_{\mathfrak f^-(e)}g(\eta),
	\end{equation}
	for all edges $ e\in E_N $ and $ \eta\in\Sigma_N $. We will see that for these systems the hydrodynamics can be treated with the same method of gradient systems. In particular later in section \eqref{sec:slv} we study the scaling limits of systems where gradient  and  circulation dynamics are superposed. Now we introduce some examples of this kind.

	\begin{example}[\emph{A non gradient lattice gas with local decomposition}]\label{exa:SSEP+vor}
		We construct a model of particles satisfying an exclusion rule, with jumps only trough nearest neighbour sites and having a non trivial decomposition of the instantaneous current \eqref{hodgefun2} with $C^{(i)}=0$ and $h$ and $g$ local functions. The functions $h$ and $g$ have to be chosen suitably in such a way that the instantaneous current is always zero inside cluster of particles and empty clusters and has to be always such that $j_\eta(x,y)\geq 0$ when $\eta(x)=1$ and $\eta(y)=0$. A possible choice is the following perturbation of the SEP. We fix $h(\eta)=-\eta(0)$ and $g(\eta)$ with $D(g)=\{0, e^{(1)},e^{(2)}, e^{(1)}+e^{(2)} \}$ (we denote by $0$ the vertex $(0,0)$) defined as follows. We have $g(\eta)=\alpha$ if $\eta(0)=\eta(e^{(1)}+e^{(2)})=1$ and
		$\eta(e^{(1)})=\eta(e^{(2)})=0$. We have also  $g(\eta)=\beta$ if $\eta(0)=\eta(e^{(1)}+e^{(2)})=0$ and
		$\eta(e^{(1)})=\eta(e^{(2)})=1$. The real numbers $\alpha,\beta$ are such that $|\alpha|+|\beta|<1$. For all the remaining configurations we have $g(\eta)=0$.  Since $\Sigma=\{0,1\}$ the rates of jump are uniquely determined by $c_{x,y}(\eta)=\left[j_\eta(x,y)\right]_+$.
	\end{example}

	\begin{example}[\emph{A perturbed zero range dynamics}]\label{exa:ZR+vor}
		A face  
		$\mathfrak f=\{0, e^{(1)}, e^{(2)}, e^{(1)}+e^{(2)}\}$ is occupied in the configuration
		$\eta\in \mathbb N^{V_N}$ if $ \eta(x) \neq 0$ for some $ x\in \mathfrak{f} $.
		Consider two non negative functions $w^\pm$ that are identically zero when the face $\mathfrak{f}$
		is not occupied. Given a positive function $\tilde h:\mathbb N\to \mathbb R^+$, we define the rates of jump as
		\begin{equation}\label{chiu}
		c_{e^-,e^+}(\eta)=\tilde h(\eta(e^-))+\tau_{\mathfrak f^+(e)}w^++\tau_{\mathfrak f^-(e)}w^-\,.
		\end{equation}
		This corresponds to a perturbation of a zero range dynamics such that one particle jumps from one site with $k$ particles with a rate $\tilde h(k)$. The perturbation increases the rates of jump if the jump is on the edge of a full face. The gain depends on the orientation and the effect of different faces is additive. For such a model the instantaneous current
		has a local decomposition \eqref{hodgefun2} with $h(\eta)=-\tilde h(\eta(0))$ and $g(\eta)=w^+(\eta)-w^-(\eta)$.
	\end{example}
	
	The decomposition can be extended to higher dimensions. For the three dimensional case we refer to \cite{thesis}.
	
	\section{Interacting particle systems  with vorticty}\label{sec:Ipsv}

	The models presented in examples \ref{exa:SSEP+vor} and \ref{exa:ZR+vor} are  superposition of a gradient system and a circulation one, see definition \eqref{eq:cicrtype}. This kind of models are not gradient along the classical definition. Here we want to study them from the microscopic point of view and giving some physical motivation why we talk about them as \emph{interacting particle systems with vorticity}, this will become more clear at the end of section \ref{sec:SL}. A better discussion with graphical examples will appear in \cite{DGG}.
	
	Let us consider  the instantaneous current \eqref {eq:istcur} with a decomposition \eqref{hodgefun2} as 
	\begin{equation}\label{eq:localHod2}
	j_\eta(x,y)=[\tau_y h(\eta)-\tau_x h(\eta)]+[\tau_{f^+(x,y)}g(\eta)-\tau_{f^-(x,y)}g(\eta)]=j^h_\eta(x,y)+j^g_\eta(x,y),
	\end{equation}
	with $ h$ and $g $ local functions. We are defining $ j^h_\eta(x,y):=\tau_y h(\eta)-\tau_x h(\eta) $ and $ j^g_\eta(x,y):=\tau_{f^+(x,y)}g(\eta)-\tau_{f^-(x,y)}g(\eta) $.  For example, taking an exclusion process with rates 
	\begin{equation}\label{eq:vortrate}
	c_{x,y}(\eta)=\eta(x)(1-\eta(y))+\eta(x)[\tau_{f^+(x,y)}g(\eta)-\tau_{f^-(x,y)}g(\eta)],
	\end{equation}
	we have  $ j_\eta(x,y)$ as in \eqref{eq:localHod2} with $ h(\eta)=-\eta(0) $, note that example \ref{exa:SSEP+vor} is of this form.
	
	Models with $ j_\eta(x,y) $ as in \eqref{eq:localHod2} can be thought as a generalization of the gradient case $ j_\eta(x,y)=[\tau_y h(\eta)-\tau_x h(\eta)] $, indeed the current is a gradient part plus an \emph{orthogonal gradient part} (discrete bidimensional curl).  Because of the presence of this discrete  curl we use the terminology of "\emph{exclusion process with vorticity}". 
	
	When the rates satisfies \eqref{eq:localHod2}, we will see that the hydrodynamics for the empirical measure \eqref{eq:empmisxi} works exactly as if only the gradient part was present because
	\begin{equation}\label{eq:divg}
	\nabla\cdot j_\eta^g(x)=0,\,\,\,\forall\,\,x\in V_N,
	\end{equation}
	that is the part of the dynamics related to the current $ j^g_\eta(x,y) $  does not give any macroscopic effect to the hydrodynamics of the particles density because its contribution to the microscopic conservation law \eqref{eq:numcon1} is already zero. To observe macroscopically the effect of the discrete curl we have to consider the scaling limits of the current flow  $ J_t(x,y) $ of formula \eqref{eq:currver}.  In section \ref{sec:slv}  we derive  the macroscopic current $ J(\rho) $ that will appear in the hydrodynamics
	$
	\partial_t\rho = \nabla\cdot (-J(\rho))
	$.
	Another physical phenomena of this kind of dynamics \eqref{eq:vortrate} is that  they  are diffusive even if  in general they are  not reversible on the torus $ \mathbb{T}^n_N $, namely this means that at the stationary state there is a non-zero macroscopic current  \eqref{eq:asJ=asj}. For an explicit example see \cite{DGG}.
	
	\section{Scaling limits and transport coefficients of diffusive models}\label{sec:SL}
	
	To derive the hydrodynamics of diffusive systems the rates are multiplied by a factor $ N^2 $ (diffusive time scale) and the space scale $ \varepsilon=1/N $ is considered. The particles jump on the discrete torus $ \mathbb T^n_\varepsilon:=\varepsilon\mathbb Z/\mathbb{Z} $ with mesh of size $ \varepsilon $. When $ N $ goes to infinity $ \mathbb T^n_\varepsilon $ approximates the continuous torus $ \mathbb T^n=[0,1)^n $. A very general class of diffusive systems are models that are reversible with respect to a Gibbs measure when no boundary conditions are imposed.  Reversibility with respect to a measure $ \mu_N $ means $\langle f,\mathcal{L}_N g \rangle_{\mu_N}=\langle \mathcal{L}_N f, g\rangle _{\mu_N} $ for all functions $ f,g $ while stationarity means$  \langle \mathcal L _N f \rangle_{\mu_N}=0 $.  $ \langle \cdot \rangle $ is the expectation on $ \Sigma_N $ respect to $ \mu_N $ and $  \langle \cdot,\cdot  \rangle_{\mu_N} $ is the scalar product respect to $ \mu_N $. We assume $ \mu_N $ to be a grand-canonical measure parametrized by the density $ \rho $, i.e. $ \mathbb{E}_{\mu_N}(\eta(x))=\rho $. For this reason instead of $ \mu_N $ we are going to use the notation $ \mu_N^\rho $.
	
	The macroscopic evolution of the mass is described by the \emph{empirical measure}. This is a positive measure on the continuous torus $\mathbb T^n$ associated to any fixed microscopic configuration $\eta$,  defined as a convex combination of delta measures 
	\begin{equation}\label{eq:empmisxi}
	\pi_N(\eta):=\frac 1N\sum_{x\in V_N}\eta(x)\delta_x\,.
	\end{equation}
	It represents a mass density or an energy density along the interpretation of the model. Integrating a continuous function $f:\mathbb T^n\to \mathbb R$ with respect to $\pi_N(\eta)$ we get
	$
	\int_{\mathbb T^n}f\,d\pi_N(\eta)=\frac 1N\sum_{x\in V_N}f(x)\eta(x)\,.
	$
	In the hydrodynamic scaling limit the empirical measure becomes deterministic and absolutely continuous for suitable initial conditions $\xi_0$ associated to a given density profile $\gamma(x)dx$, in the sense that in probability
	\begin{equation}\label{eq:asc}
	\lim_{N\to +\infty}\int_{\mathbb T^n} f\, d\pi_N(\xi_0)=\int_{\mathbb T^n} f(x)\gamma(x)dx.
	\end{equation}
	Let  $  P_N^{\gamma} $ be the distribution of the Markov chain of the energy-mass/particle interacting model  with initial condition associated to $ \gamma $ as in \eqref{eq:asc} . On $ D([0,T];\mc M^1(\mathbb T^n)) $  the space of trajectories from $ [0,T] $ to the space of positive measure $ \mc M^1(\mathbb T^n) $,  $ \bb P_N^\gamma:= P^{\gamma}_N\circ\pi_N^{-1} $ is the measure induced by the empirical measure.
	We have that $\pi_N(\eta_t)$ is associated to the density profile $\rho(x,t)dx$ where $\rho$ is the weak solution to a diffusive equation
	with initial condition $\gamma$, i.e.
	$
	\bb P^\gamma_N\os{d}{\us{N}{\longrightarrow}}\delta_\rho
	$
	where and  $ \delta_{\rho} $ is the distribution concentrated on the unique weak solution of  a Cauchy problem 
	\begin{equation}\label{eq:drc}
	\left\{
	\begin{array}{l}
	\partial_t\rho=\nabla\cdot(D(\rho)\nabla\rho)\\
	\rho(x,0)=\gamma(x).\\
	\end{array}
	\right.
	\end{equation}
	This is a space-time law of large numbers, where $ D(\sigma) $ is a positive symmetric matrix called \emph{diffusion matrix}.

	\subsection{Qualitative derivation of hydrodynamics }
	
	In this subsection we illustrate the general structure of the proof of the hydrodynamic limit for reversible gradient models  on the torus $ \mathbb{T}^n_\varepsilon $.   
	We use the notion $ \xi $ of section \ref{sec:en models} of energy-mass models because for them we gave some example of weakly asymmetric model and   we want to emphasize that the KMP model is gradient. But the whole scheme apply to particle models in the same way.

	The starting point for the hydrodynamic description is the continuity equation
	\begin{equation}\label{disc-cont}
	\xi_t(x)-\xi_0(x)=-\nabla\cdot \mathcal J_t(x)\,,
	\end{equation}
	where $\mathcal J_t$ has been defined in subsection \ref{subsec:Ie-mc} and $\nabla\cdot $ denotes the discrete divergence defined in \eqref{eq:disdiv}. Using \eqref{eq:massmart} we can rewrite \eqref{disc-cont} as \eqref{eq:masscon} with $ \mathbb{F}=0 $. 
	Multiplying  \eqref{eq:masscon}  by a test function $\psi$, dividing by $N$ and summing over $x$  we obtain
	\begin{equation}\label{gate}
	\int_{\mathbb{T}^n} \psi \,d\pi_N(\xi_t)-\int_{\mathbb{T}^n} \psi\, d\pi_N(\xi_0)=-N\int_0^t\sum_{x\in V_N}\nabla\cdot j_{\xi_s}(x)\,\psi\left(x\right)\,ds+o(1)\,.
	\end{equation}
	The infinitesimal term $ o(1) $ comes from the martingale term. The idea is that the martingales $ M_t(x,y) $ in \eqref{eq:mart} describe some microscopic fluctuations whose additive macroscopic contributions  vanishes as $ N\rightarrow \infty $ as they are mean zero martingales and are almost independent for different bonds. This contribution   can be shown
	to be negligible (in probability) in the limit of large $N$ with the methods of \cite{Gonc,KL99}. Using the gradient condition $ j_{\xi}(x,y)=\tau_x h(\xi)-\tau_y h(\xi)$, for example for the KMP \eqref{eq:KMPbulkgen}  and KMPd \eqref{eq:KMPdrate} we have $ h(\xi)=\frac {\xi(0)} 2 $,  and  performing a double discrete integration by part, up to the infinitesimal term, one has that the right hand side of \eqref{gate} is
	$
	\frac{1}{N}\sum_{x\in V_N}\int_0^t \tau_xh(\xi_{{s}})\left[N^2\left(\psi\left(x+\frac{1}{N}\right)+\psi\left(x-\frac{1}{N}\right)-
	2\psi\Big(x\Big)\right)\right]\,ds\,.
	$
	Considering a $C^2$ test function $\psi$, the term inside squared parenthesis 
	coincides with $\Delta\psi\left(x\right)$ up to an uniformly infinitesimal term.
	
	At this point the main issue in proving hydrodynamic behavior is to prove the validity of a local equilibrium property. Let us define
	\begin{equation}\label{local}
	A(\rho)=\mathbb E_{\mu^\rho_{N}}\left(h (\xi)\right)\,,
	\end{equation}
	where  $\mu^\rho_N $ is the invariant measure characterized by a density profile $ \rho  $, that is $\mathbb{E}_{\mu^\rho_N} (\xi(x))=\rho $.
	The local equilibrium property is explicitly stated through a replacement lemma that shows that (in probability)
	\begin{equation}\label{eq:repla}
	\frac{1}{N}\sum_{x\in V_N}\int_0^t \tau_xh(\xi_{{s}})\Delta\psi\left(x\right)\,ds\simeq
	\frac{1}{N}\sum_{x\in V_N}\int_0^t A\left(\frac{\int_{B_{x}}d\pi_N(\xi_s)}{|B_{x}|}\right)\Delta\psi\left(x\right)\,ds
	\end{equation}
	where $B_{x}$ is a microscopically large but macroscopically small volume around
	the point $x\in V_N$. For a precise formulation of \eqref{eq:repla} see lemma 1.10 and corollary 1.3 respectively  in chapter 5 and  in chapter 6 of \cite{KL99} or chapter 2 in \cite{Gonc}.
	This allows to write (up to infinitesimal corrections) equation \eqref{gate} in terms only of the empirical measure.
	Substituting the r.h.s. of \eqref{eq:repla} in the place of the r.h.s. of \eqref{gate}, we obtain that in the limit of large $N$ the empirical measure $\pi_N(\eta_t)$ converges in weak sense to $\rho(x,t)dx$ satisfying for any $ C^2 $ test function $\psi$
	\begin{equation}\label{widro}
	\int_{\mathbb{T}^n} \psi(x)\rho(x,t)\,dx-\int_{\mathbb{T}^n}\psi(x)\rho(x,0)\,dx=\int_0^t ds\int_{\mathbb{T}^n} A(\rho(x,s))\Delta\psi(x)\,dx\,.
	\end{equation}
	Equation \eqref{widro} is a weak form of \eqref{eq:drc} with diagonal diffusion matrix $ D(\rho) $ with each term in the diagonal equal to
	$
	D(\rho)= \frac{ d A(\rho)}{d \rho}\,.
	$
	We are calling $ D(\rho) $ both the number and the diagonal matrix $ D(\rho)\mathbb I $. For  $h(\xi)=\frac{\xi(0)}{2}$ it is $A(\rho)=\frac{\rho}{2}$. To have an unitary diffusion matrix we multiply all the rates of transition by a factor of 2 and correspondingly the diffusion matrix is the identity matrix.
	Equation \eqref{eq:drc} can be written in the form
	\begin{equation}\label{eq:hdconJ}
	\partial_t\rho+\nabla \cdot (J(\rho))=0, \text{ with } J(\rho)= -D(\rho)\nabla\rho,
	\end{equation}
	where the macroscopic current $ J(\rho) $ associated to  $ \rho  $  satisfies the \emph{Fick's law}.The hydrodynamics for weakly asymmetric diffusive models of subsection \eqref{subsec:WAm} is
	\begin{equation}\label{hydrog}
	\partial_t\rho=\nabla\cdot\left(-J_E(\rho)\right)\text{  with }J_E(\rho):=D(\rho)\nabla\rho-\sigma(\rho)E.
	\end{equation}
	The positive definite matrix $\sigma$ is called the \e{mobility}. For  the weakly asymmetric versions of the KMP and the KMPd, in subsection \ref{subsec:WAm} it is respectively $ \sigma(\rho)=2 \mathbb E_{\mu^\rho_N}\left[g(\eta)\right]= \rho^2$ and $ \rho+\rho^2 $, where respectively $  g(\xi)=\frac 16\big(\xi(0)^2+\xi(1/N)^2-\xi(0)\xi(1/N)\big)$ and $  g(\xi)={12}\big(2\xi(x)^2+2\xi(y)^2-2\xi(x)\xi(y)+3\xi(x)+3\xi(y)\big) $. For a discussion on the computations of these kind of expectations see \cite{BGL}.
	
	The hydrodynamics was derived with periodic boundary conditions but in the bulk  it is   still the same for  a  boundary driven  version of the system, see \cite{ELS96}.

	\section{Scaling limit of an exclusion process with vorticity}\label{sec:slv}
	
	In this section  we want to show how to compute the scaling limit of the macroscopic current $ J(\rho) $ for diffusive models  with vorticity of section \ref{sec:Ipsv}, namely having  the instantaneous current with an expression like \eqref{eq:localHod2}. Here for the purpose of the paper the treatment will be qualitative.  It is the first time that the hydrodynamics of this kind of models is discussed. A complete rigorous treatment of the problem is now being developed in the work in progress\cite{DGG}, where a generalized picture of the Fick's law is under construction. Here we will discuss its main ideas. We consider  a discrete torus of mesh $ \varepsilon=1/N $ but specifically in dimension 2, i.e. $V_N=\mathbb{T}^2_\varepsilon $.

	If the current has an Hodge decomposition \eqref{hodgefun2} only the gradient part  contributes to the hydrodynamics \eqref{gate}, indeed $ \nabla \cdot j_\eta (x)=\nabla \cdot j^h_\eta (x)$ since $ \nabla\cdot j^g_\eta (x)=\nabla\cdot j^H_\eta (x)=0 $ with $ j^H_\eta (x,y)=C^1(\eta)\phi^1(x,y)+C^2(\eta)\phi^2(x,y) $. So if  the gradient part of the current  $ j^h_\eta (x) $ is  diffusive with respect to a local gradient function $ h(\cdot) $,  the hydrodynamics of $ \pi_N(\eta) $ works exactly as if we considered  a model with $ j_\eta (x,y)=j_\eta^h(x,y) $ along the scheme in section \ref{sec:SL}.

Now we want to study the scaling limits of the current $ J(\rho) $ appearing in \eqref{eq:hdconJ}, as model of reference for what we are going to present, the reader should keep in mind the exclusion process of example \ref{exa:SSEP+vor}  but with $ \alpha=\beta $.   More precisely the model  has the rates \eqref{eq:vortrate} with the local function $ g(\eta) $ defined as 
	\begin{equation}\label{eq:defg}
	g(\eta):=\left\{
	\begin{array}{ll}
	\alpha & \textrm{if}\ \eta(0)=\eta(\frac{e^{(1)}}{N}+\frac{e^{(2)}}{N})=1\ \textrm{and}\ \eta(\frac{e^{(1)}}{N})=\eta(\frac{e^{(2)}}{N})=0\,,\\
	\alpha & \textrm{if}\ \eta(\frac{e^{(1)}}{N})=\eta(\frac{e^{(2)}}{N})=1\ \textrm{and}\ \eta(0)=\eta(\frac{e^{(1)}}{N}+\frac{e^{(2)}}{N})=0\,,\\
	0 & \textrm{otherwise}\,,
	\end{array}
	\right.
	\end{equation}
	where $\alpha$ is a real parameter such that $|\alpha|<1$. 
	The informal and intuitive description of the dynamics associated to the rates \eqref{eq:vortrate} is the following. Particles perform a  simple exclusion process, but the faces containing exactly 2 particles located at sites which are not nearest neighbors  let the particles rotate anticlockwise when $\alpha>0$
	and clockwise when $\alpha<0$ with a rate equal to $|\alpha|$.  For this choice of the parameters, the model of example \ref{exa:SSEP+vor} can be proven to be a  non-reversible stationary dynamics with respect to Bernoulli measures of density parameter $ \rho  $. In this section, the language will be general for   a model that is invariant with respect to a measure $ \mu^\rho_N $ parametrized by a density $ \rho $, having the decomposition \eqref{eq:localHod2} and  hydrodynamics for the empirical measure of the form \eqref{eq:hdconJ}, while the results will be made explicit for the toy model \eqref{eq:defg}.

	The scaling limits for the current $ J(\rho) $ it is obtained from the \emph{empirical current measure} $ \mathbb{J}_N$ in the space of the vector signed measure $\mathcal M (\mathbb{T}^2,\mathbb{R}^2) $ defined as  
	
	\begin{equation}\label{eq:empcurrdef}
	\int_{\mathbb{T}^2} H\cdot d\mathbb{J}_N:=\frac{1}{N^2}\underset{\{x,y\}\in\mathcal{E}_N}{\sum}J_t(x,y)\mathbb{H}(x,y) \text{ where  } \mathbb{H}(x,y)=\int^y_x H(z)\cdot dz. 
	\end{equation}
	The family $ (\mathbb{J}_N(t))_{t\in[0,T]} $ belongs to the space $ D([0,T],\mathcal M (\mathbb{T}^2,\mathbb{R}^2)) $ of trajectories from $[0,T]  $ to $ \mathcal M (\mathbb{T}^2,\mathbb{R}^2) $. Calling $ \mathbb{P}_{\mathbb{J}_N}:= P_N^\gamma\circ \mathbb{J}^{-1}_N $ the measure  induced by empirical current measure on $ D([0,T],\mathcal M (\mathbb{T}^2,\mathbb{R}^2) )$, we have that   $ \mathbb{J}_N(t) $ is associated to a vector signed measure $ J(\rho)dx $ in weak sense, that is in probability  for any $ C^1 $ vector field on $ \mathbb{T}^2 $ we will have
	\begin{equation}\label{idroJ}
	\begin{array}{ll}
	&\underset{{N\to+\infty}}{\lim}\int_{\Lambda}H\cdot d\mathbb J_N(t)=\int_{\Lambda}dx \int_0^Tdt J(\rho_t(x))\cdot H(x), \\
	\\ &J(\rho)=-D(\rho)\begin{pmatrix} 1 & 0\\
	 0 & 1 \end{pmatrix}\nabla\rho-D^\perp(\rho)\begin{pmatrix} 0 & -1\\ 1 & 0 \end{pmatrix}\nabla\rho\,,
	\end{array}
	\end{equation}
	where $ D(\rho) $ and $ D^\perp(\rho) $ are two real coefficients depending on $ \rho $, $\rho_t(x)$ is the solution of the Cauchy problem \eqref{eq:drc} and $ J(\rho(0)) $ is equal 0 by definition.  This means that $ \mathbb{P}_{\mathbb{J}_N} \os{d}{\us{N}{\longrightarrow}}\delta_{J(\rho)} $ where  $ \delta_{J(\rho)} $ is the distribution concentrated on the measure $ J(\rho)dx $ that we have just described.  For the model \eqref{eq:defg} we will show that $ D(\rho)=1 $ and $ D^\perp(\rho)= \frac{d}{d\rho} (2\alpha(\rho-\rho^2)^2) $.
	The derivation of the hydrodynamics starts from the martingale
	\begin{equation}\label{eq:martTotcorr}
	M(t)= \frac {1}{N^d}\sum_{\{x,y\}\in\mathcal{E}_N}J_t(x,y)\mathbb H (x,y) - N^{2-d} \int ^t_0 ds \,\, \sum_{\{x,y\}\in\mathcal{E}_N} j_{\eta_s}(x,y) \mathbb{H}(x,y),
	\end{equation} 
	where $ N^2 $ is the diffusive scaling and the factor $N^{-d}$  it is a normalization.
	By the antisymmetry of the  discrete vector fields there is no ambiguity in this definition. Therefore 
	$
	\frac {1}{N^d}\sum_{\{x,y\}\in\mathcal{E}_N}J_t(x,y)\mathbb H (x,y)$$=N^{2-d} \int ^t_0 ds \,\, \sum_{\{x,y\}\in\mathcal{E}_N} j_{\eta_s}(x,y) \mathbb{H}(x,y)  + o(1),
	$
where $ o(1) $ is  a negligible (in probability) martingale term for large $ N $, for which holds a discussion like that one about the martingale in \eqref{gate}. From   \eqref{eq:localHod2} 
	{\small \begin{equation}\label{darimpiaz}
	\sum_{\{x,y\}\in\mathcal{E}_N} j_{\eta_s}(x,y) \mathbb{H}(x,y)=\int_0^t\left[\sum_{x\in V_N}\tau_x h(\eta)(x)\nabla\cdot \mathbb H(x) +\sum_{\mathfrak f\in \mathcal F_N}\tau_{\mathfrak f}g(\eta_s)\sum_{(x,y)\in f^\circlearrowleft}\mathbb{H}(x,y)\right],
	\end{equation}}
where $ N^2\nabla\cdot\mathbb{H}(x)=\nabla\cdot H(x)+o(1/N) $ and $ N^2\sum_{(x,y)\in f^\circlearrowleft}\mathbb{H}(x,y)=\nabla^{\perp}\cdot H(z)+o(1/N) $.
	In the above formula $z$ is any point belonging to the face,
	while given a $C^1$ vector field $H=(H_1,H_2)$ we used the notation
	$\nabla^\perp\cdot H(z):=-\partial_yH_1(z)+\partial_xH_2(z)$.
	When $ N $ is diverging, we assume the local equilibrium hypothesis with respect to the grand-canonical measure $ \mu_N^\rho $ to prove with a replacement lemma as discussed in section\ref{sec:SL}, this means that  \eqref{darimpiaz} converges (in probability) to
	\begin{equation}\label{wHDJ}
	\int_0^t ds\int_{\mathbb T^2} dx\,\left[a(\rho(x,s))\nabla\cdot H(x)+a^\perp(\rho(x,s))\nabla^\perp\cdot H(x)\right],
	\end{equation}
	applying 
	the replacement lemmas $ a(\rho)=$$\mathbb E_{\mu^\rho_N}[h(\eta)] $ and $ a^\perp(\rho)=$$\mathbb E_{\mu^\rho_N}[g(\eta)]  $,  for the model of reference \eqref{eq:defg} its $ a(\rho)=\rho $ and $ a^\perp(\rho)= 2\alpha[\rho(1-\rho)]^2 $.
	Formula \eqref{wHDJ} is a weak form of $\int_0^tds\int_\Lambda J(\rho)\cdot H\,dx$ with
	\begin{equation}\label{tipJ}
	J(\rho)=-\nabla a(\rho)-\nabla^\perp a^\perp(\rho)=-D(\rho)\nabla\rho - D^\perp(\rho)\nabla^\perp\rho,
	\end{equation}
	where  $ D(\rho)=d (a(\rho))/d\rho $,  $ D^\perp(\rho)= d(a^\perp(\rho))/d\rho  $ and  $ \nabla^\perp f:=(-\partial_y f,\partial_x f) $.  As we expected from the microscopic argument \eqref{eq:divg} we have 
	$
	\nabla\cdot (-D(\rho)\nabla\rho - D^\perp(\rho)\nabla^\perp\rho)=\nabla\cdot(-D(\rho)\nabla\rho),
	$
	hence the hydrodynamics is left unchanged with respect to the usual gradient case.  For the model \eqref{eq:defg} we obtain $ D(\rho)=1 $  and $ D^\perp(\rho)= \frac{d}{d\rho} \ (2\alpha(\rho-\rho^2)^2)$ as anticipated above. 
	Hence formula \eqref{tipJ} can be rewritten 
	in the form
	\begin{equation}\label{tipJ2}
	J(\rho)=-\left[D(\rho)\begin{pmatrix} 1 & 0\\ 0 & 1 \end{pmatrix}+D^\perp(\rho)\begin{pmatrix} 0 & -1\\ 1 & 0 \end{pmatrix}\right]\nabla\rho.
	\end{equation}
	From the computations presented here and other considerations in \cite{DGG} we think that  the Fick's law \eqref{eq:hdconJ} for particle models has to be replaced by the general picture
	\begin{equation}\label{eq:JasymD}
	J(\rho)=-\mathcal D(\rho)\nabla\rho, 
	\end{equation}
	where the diffusion matrix $ \mathcal D (\rho)$ is  positive but not necessarily symmetric and  in general with respect to \eqref{tipJ2} the terms on the same diagonal can have different coefficients.
	
	\begin{remark}
		An important question is to understand if there exist models with vorticity that are also reversible. At the present stage, we are not able neither to find reversible models of this kind neither to prove that this property is a genuine microscopic non equilibrium property, if this is case, it looks that typically they will be diffusive models with a non-zero average microscopic current at the stationary state. 
	\end{remark}

	\section{Green-Kubo's formula and perspectives in non-gradient particles systems}\label{sec:GK}
	
	Scaling limits of non gradient particles systems can be proved to be diffusive with the methods developed in \cite{Q,Var94} if the spectral gap of the generator satisfies suitable conditions  \cite{KL99}, but   even in one dimension when the instantaneous current is not gradient there are no explicit PDEs. For what we know, the only case where there is an explicit  PDE is \cite{Wick}, where the author consider a spatial inhomogeneous simple symmetric exclusion process where particles jump with two different constant along an  edge is even or odd. The model of Wick is translational covariant, see \eqref{imbr},  with respect to translations on two sites instead of one. To look at  Wick model into the context of this paper and in particular of this section, one has to generalize the decomposition \eqref{imbr} for translational covariant models on two sites, this is done considering a  renormalized current on two sites and rewriting the decomposition for it. Then, taking a lattice with a even number of sites, the discrete hydrodynamics \eqref{gate} can be written with respect to this current. 
	
	 We  start to explore if the decompositions \eqref{imbr} and \eqref{hodgefun2} can tell something about this problem. Let us consider exclusion processes in one dimension  on the torus  with nearest neighbours interaction, reversible with respect to $ \mu^\rho_N $  and non gradient, a case is example \ref{exa:diffSEP}. For these models the hydrodynamics is expected to be diffusive with the diffusion coefficient  having the following   variational expression 
	\begin{equation}\label{eq:var}
	D(\rho)=\frac{1}{2\chi(\rho)}\,\underset{f}{\inf}\,\,\mathbb{E}_{\mu^\rho_N}\big[c_{0,1}(\eta)  \big(  (\eta(0)-\eta^{0,1}(0))+   \sum_{x\in V_N} (S_{0,1}  \tau_x )f(\eta)\big) ^2 \, \big],
	\end{equation}
	where $ \chi(\rho) $ is the mobility $ \mathbb{E}_{\mu^\rho_N}(\eta^2(0))-\rho^2 $, $ (S_{x,y} f)= f(\eta^{x,y})-f(\eta) $ and the $\inf $ is  over all functions $ f:\Sigma_N\rightarrow \mathbb{R}$.  
	This  is discussed in chapter 2 of part 2 in \cite{Spohn} and has been proved for the 2-SEP in chapter 7 of \cite{KL99}. The variational formula \eqref{eq:var} is proved to be equal to the  \emph{Green-Kubo's} formula  for interacting particles systems
	\begin{equation}\label{eq:GK}
	D(\rho)=\frac{1}{2\chi(\rho)}\big[ \mathbb{E}_{\mu_N}(c_{0,1}(\eta)) - 2\int^{+\infty}_0\mathbb{E}_{\mu_N} (j_\eta (0,1)e^{\mathcal{L}_Nt}\tau_x j_\eta(0,1))\big],
	\end{equation}
	where $ e^{\mathcal{L}_Nt} $ is the evolution operator  of the Markov process.  We consider translational covariant rates (remark \ref{re:trasrate}) to have the decomposition \eqref{imbr}, plugging this one   in \eqref{eq:GK} we find
	\begin{equation}\label{eq:GK2}
	D(\rho)=\frac{1}{2\chi(\rho)}\big[ \mathbb{E}_{\mu_N}(c_{0,1}(\eta)) - 2 N\,\mathbb{E}_{\mu_N} (C(\eta)\mathcal{L}_N^{-1} C(\eta))\big],
	\end{equation}
	where $ \mathcal{L}_N^{-1} $ is the generalized inverse operator of $ \mathcal{L}_N $  (for $ f(\eta) $ constant function $ \mathcal{L}_N f(\eta)=0 $), for this definition see \cite{Hunter}.  Formula \eqref{eq:GK2} tells us that just the harmonic part of the current $ C(\eta) $ contributes to the second term of the Green-Kubo's formula and  for gradient systems we have   $ D(\rho)=\frac{1}{2\chi(\rho)} \mathbb{E}_{\mu_N}(c_{0,1}(\eta)) $ even if the $ h  $ is not local, admitting that such models exist.  Expression \eqref{eq:GK2} has an equivalent variational formulation with a minimizer that is computable in principle. The term $ \mathbb{E}_{\mu^\rho_N} (C(\eta)\mathcal{L}_N^{-1} C(\eta)) $ can be seen as the scalar  product $ \langle C(\eta),\mathcal{L}_N^{-1} C(\eta)\rangle_{\mu^\rho_N} $, where $ \langle f, g \rangle_{\mu^\rho_N}= \sum_{\eta} f(\eta)g (\eta)\mu^\rho_N(\eta) $. Since $ \mathcal{L}_N $ is symmetric with respect to this scalar product we have 
	\begin{equation}\label{eq:var2}
	\langle C(\eta),\mathcal{L}_N^{-1} C(\eta)\rangle_{\mu^\rho_N}=\underset{f}{\inf}\left\{ - \langle f, \mathcal{L}_N f\rangle_{\mu^\rho_N} -2 \langle C(\eta), f \rangle_{\mu^\rho_N}\right\}
	\end{equation}
	where the minimizer is over all function $  f:\Sigma_N\rightarrow \mathbb{R} $  and a solution is given by
	\begin{equation}\label{eq:minimizer}
	\mathcal{L}_N f(\eta)= - C(\eta) \text{ for all } \eta\in \Sigma_N.
	\end{equation}
	   The solution \eqref{eq:minimizer} is well posed since $ C(\eta) $ is orthogonal to the eigenspace of eigenvalue zero. This minimizer  looks to us more simple to solve than the one of the expression in \cite{Spohn}, for example interpreting the model of Wick as explained at the beginning of the section this minimum can be solved within the framework we are going to explain in next paragraphs. We think that  the solution of this minimizer  is equivalent to rewrite the discrete hydrodynamics in a form such that the only macroscopic relevant terms  are  reduced to the usual case of section \ref{sec:SL} of gradient systems. In some special non-gradient cases (as Wick \cite{Wick}) we expect a simplified scheme,  that is an exact case of  a more general scheme  briefly described at the end.
	   
	   The idea starts from the observation that a natural attempt  to solve \eqref{eq:minimizer} is to look for  a $ f(\eta) $ of the form  
	   	\begin{equation}\label{eq:sumf}
	   	f(\eta)=\underset{x\in V_N}{\sum}\tau_x g(\eta),
	   	\end{equation} 
	   	where $ g(\eta) $ is a local function. Note that the left-hand side of \eqref{eq:minimizer} is invariant by translation as it has to be.  Remark \ref{re:ja} gives  a connection between the minimizer and the conservation law leading to the  hydrodynamics, there   we discussed that in reversible non-gradient model the current can be rewritten in the form $  j_\eta(x,y)=j^h_\eta(x,y)+j^a_\eta(x,y)  $, where the single harmonic contributions are such that $ C(\eta)=\frac 1N\underset{x\in V_N}{\sum} j_\eta(x,x+1)=\frac 1N\underset{x\in V_N}{\sum} j^a_\eta(x,x+1) $. So to the part of the current denoted $ j^h_\eta(x,y) $ we can apply   the scheme of section \ref{sec:SL} with respect to a gradient function $ h(\eta) $, while it is not possible for the part $ j^a_\eta(x,y)$ . But if we are able  to find a local function $\tilde{g}(\eta) $ such that 
	   	\begin{equation}\label{eq:Wick}
	   	\mathcal{L}_N \tilde{g}(\eta)= j^a_\eta(0,1)+\tau \tilde{h}(\eta)-\tilde{h}(\eta),
	   	\end{equation} 
	   where $ \tilde{h}(\eta) $ is another local function,  then we are done both with the solution \eqref{eq:sumf} and the discrete form of the hydrodynamics \eqref{gate}.   Indeed respectively taking $ f(\eta) = \underset{x\in V_N}{\sum}\tau_x g(\eta)$ with $ g(\eta)=-\tilde{g}(\eta) $ we solve \eqref{eq:minimizer} and with the replacements \eqref{eq:Wick} of the harmonic contributions we will be able to treat the hydrodynamics. This is because considering the translations of relation \eqref{eq:Wick}, in the discrete hydrodynamics \eqref{gate} will contribute only the local function ${h'}(\eta)=-\tilde{h}(\eta) $ since with the local equilibrium hypothesis \eqref{local}  the terms $ \mathbb{E}_{\mu_N^\rho}(\mathcal{L_N}\tau_x g(\eta_t)$) will be negligible in the scaling limit as they are time derivatives and from the time integral they will give a contribution of order $ O(1/N^2) $ each one. At the end, the hydrodynamics will follow section \ref{sec:SL} with respect to the local function $ H(\eta):=h(\eta)+h'(\eta) $.
   
   \vspace{1cm}
   
   In the non-local decomposition \eqref{imbr} the part $ C(\eta) $ is divergence free and therefore will not appear in the hydrodynamics \eqref{gate}. We expect that writing this last one with respect to the non-local $ h^a(\eta) $ of the Hodge decomposition $ j^a_\eta (x,y)=\tau_y h^a(\eta)-\tau_x h^a(\eta)+C(\eta) $,   doing  the  substitution \eqref{eq:Wick} in  $ h^a = \sum_{x=1}^{N-1}\frac{x}{N}j^a_\eta(x,x+1)\,$, with proper cancellations the hydrodynamics will still reduce to the one related to $ H(\eta) $.

   \vspace{1cm}In general solving \eqref{eq:minimizer} with an $ f(\eta) $ of the form \eqref{eq:sumf} with the property \eqref{eq:Wick} will be not possible. But   for models like example \ref{exa:diffSEP}, we expect  a generalization of this case  where \eqref{eq:minimizer} is solved unless of  (non-local) gradients  (which will not contribute in the computation of the scalar products in  \eqref{eq:var2}) with a solution as \eqref{eq:sumf} where $ g(\eta) $ satisfies  \eqref{eq:Wick} unless  extra terms  on the right-hand side that in probabilistic sense will be of order $ o(1/N) $ and will not contribute to hydrodynamics.
   
   Similarly to \eqref{eq:GK2},  in dimension higher than one,  the extra terms for non gradient systems will come only from the harmonic part,  i.e. $ C^1(\eta)\phi^1 (x,y) $ and  $  C^1(\eta)\phi^1 (x,y)$ in dimension two.

	\begin{acknowledgement}
		The author thanks prof. Davide Gabrielli of Universit\'{a} degli studi di L'Aquila for  a plenty of discussions and ideas  in the paper. This project has received funding from the European Research Council (ERC) under  the European Union's Horizon 2020 research and innovative programme (grant agreement   No 715734).
	\end{acknowledgement}
	
%
%
%

\end{document}